\documentclass[a4paper]{article}
\usepackage{cite, graphicx, subfigure, amsmath, appendix} 
\usepackage{setspace}
\usepackage{fullpage}
\usepackage{fancyhdr}
\usepackage{fullpage}
\usepackage{caption} 
\usepackage{hyperref}
\usepackage{multirow}
\usepackage{comment}
\usepackage{wrapfig}
\usepackage{amssymb}
\usepackage{amsfonts}
\usepackage{hyperref}

\newtheorem{lemma}{Lemma}
\newtheorem{proof}{Proof}

\begin{document}
\newcommand\bra[2][]{#1\langle {#2} #1\rvert}
\newcommand\ket[2][]{#1\lvert {#2} #1\rangle}
\title{Encoding lattice structures in Quantum Computational Basis States}
\author{Kalyan Dasgupta}
\date{%
    IBM Research, Bangalore, India\\%
    }
\maketitle
\date{}

\section{Introduction}
Lattice models or structures are geometrical objects with mathematical forms, that are used to represent physical systems. They have been used widely in diverse fields, namely, in condensed matter physics, to study degrees of freedom of molecules in chemistry and in studying polymer dynamics and protein structures to name a few \cite{dimacs_rutgers}. In this article we discuss an encoding methodology of lattice structures in computational basis states of qubits (as used in quantum computing algorithms). We demonstrate a specific use case of lattice models in protein structure prediction. A variety of lattice models have been used for protein structure prediction. Although, proteins have highly irregular structures, lattice models have been used to predict structures at a basic level. Lattice models fall under the category of coarse-grained models. The geometry considered for coarse-grained models could either have a continuous representation or have a lattice representation \cite{cg-model}). Lattice models offer significant computational speedup over other representations. Fig. \ref{fig:lattice_structs} gives some of the lattice representations.
\begin{figure}[htp]
\centering 
\includegraphics[width=5in]{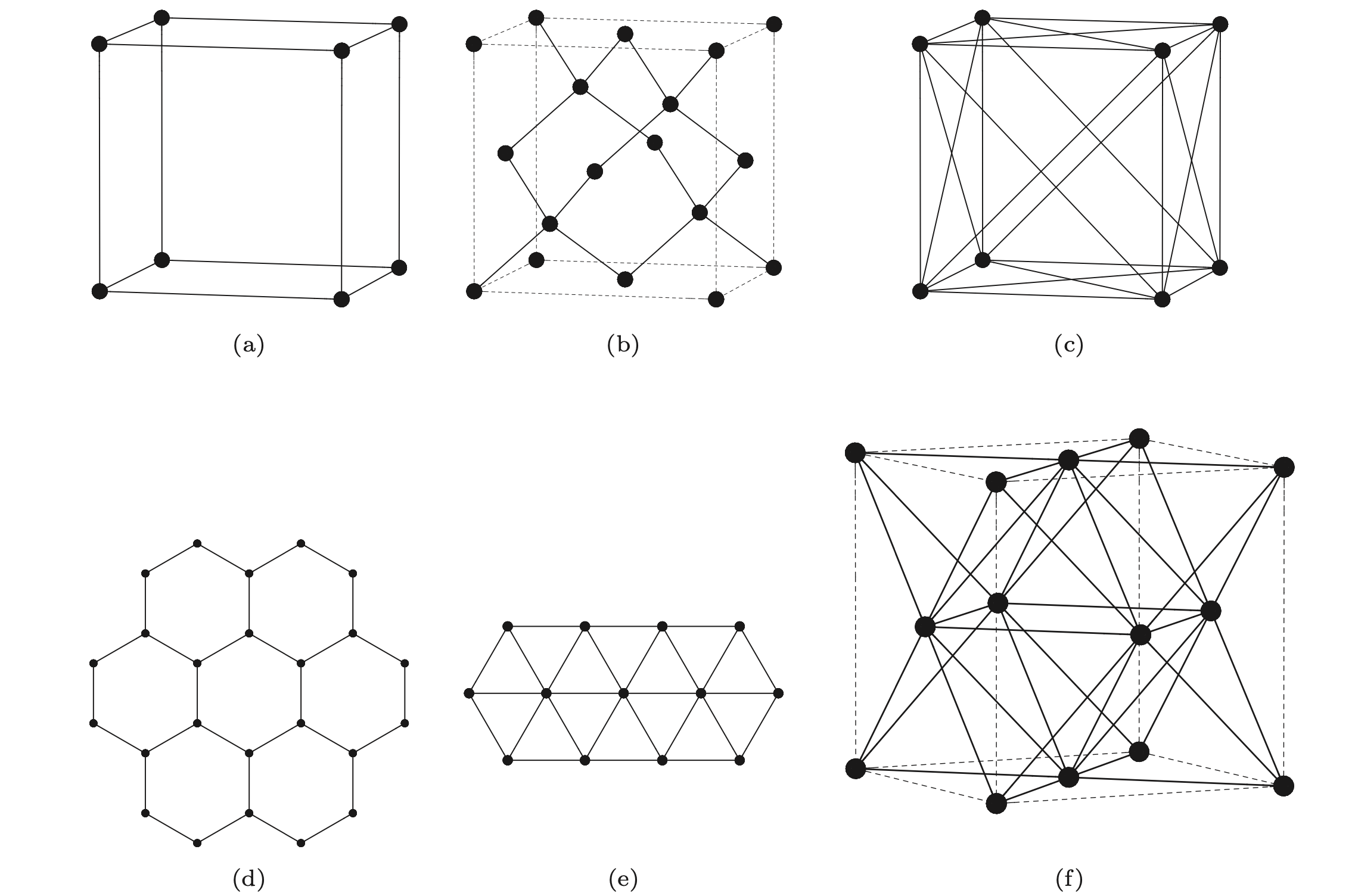} 
\caption{Lattice structures: (a) cubic, (b) diamond, (c) cubic with planar diagonals, (d) hexagonal, (e) triangular and (f) face-centred-cubic. Courtesy: \cite{dimacs_rutgers}}
\label{fig:lattice_structs}
\end{figure}

The protein structure prediction problem suffers due to the large feasible conformational space. Choosing the correct conformation from such a vast space is a big challenge. With quantum hardwares, one can create $2^n$ computational basis states form $n$ qubits. One can find several papers using quantum computing for protein structure prediction using lattice models. Cubic lattices have been used in references \cite{Babej_2018} and \cite{Perdomo}. 
In reference \cite{Anton_2021}, the authors presented a quantum model for a coarse-grained protein sequence mapped onto a tetrahedral lattice structure. In reference \cite{paper-arxiv} the authors propose a turn based encoding in a cubic lattice with larger degrees of freedom using the cubic lattice with planar and steric diagonals. 

In this article we do not propose any quantum algorithm to solve the protein structure prediction problem, instead, we propose a generic encoding methodology of lattice structures. The protein sequence, essentially, consists of turns or bonds between adjacent monomers or amino acids (also referred as a bead in coarse-grained models). The turn that the bonds can take are limited by the degrees of freedom of the lattice model being used. We show how the encoding methodology encode the turns based on the lattice model selected. We take two specific lattice models, the cubic with planar diagonals and the face-centred cubic (FCC), and demonstrate the encoding methodology. The article is organized as follows. In section 2, we discuss the lattice models and the coordinate space that the turns can take. In section 3 we discuss how the coordinates are mapped to the qubit computational basis states. Finally, we summarize the methodology in the conclusions. 

\section{Lattice models and coordinate space}
Fig. \ref{fig:fcc_lattice} (a) gives a section of the structure of the FCC lattice. Given a protein bead at the origin, we would like to know the possible set of 3D coordinates that the next bead could take following a turn. Fig. \ref{fig:fcc_lattice} (b) gives the possible coordinates in a given plane $y-z$. 
\begin{figure}[htp]
\centering 
\includegraphics[width=6in]{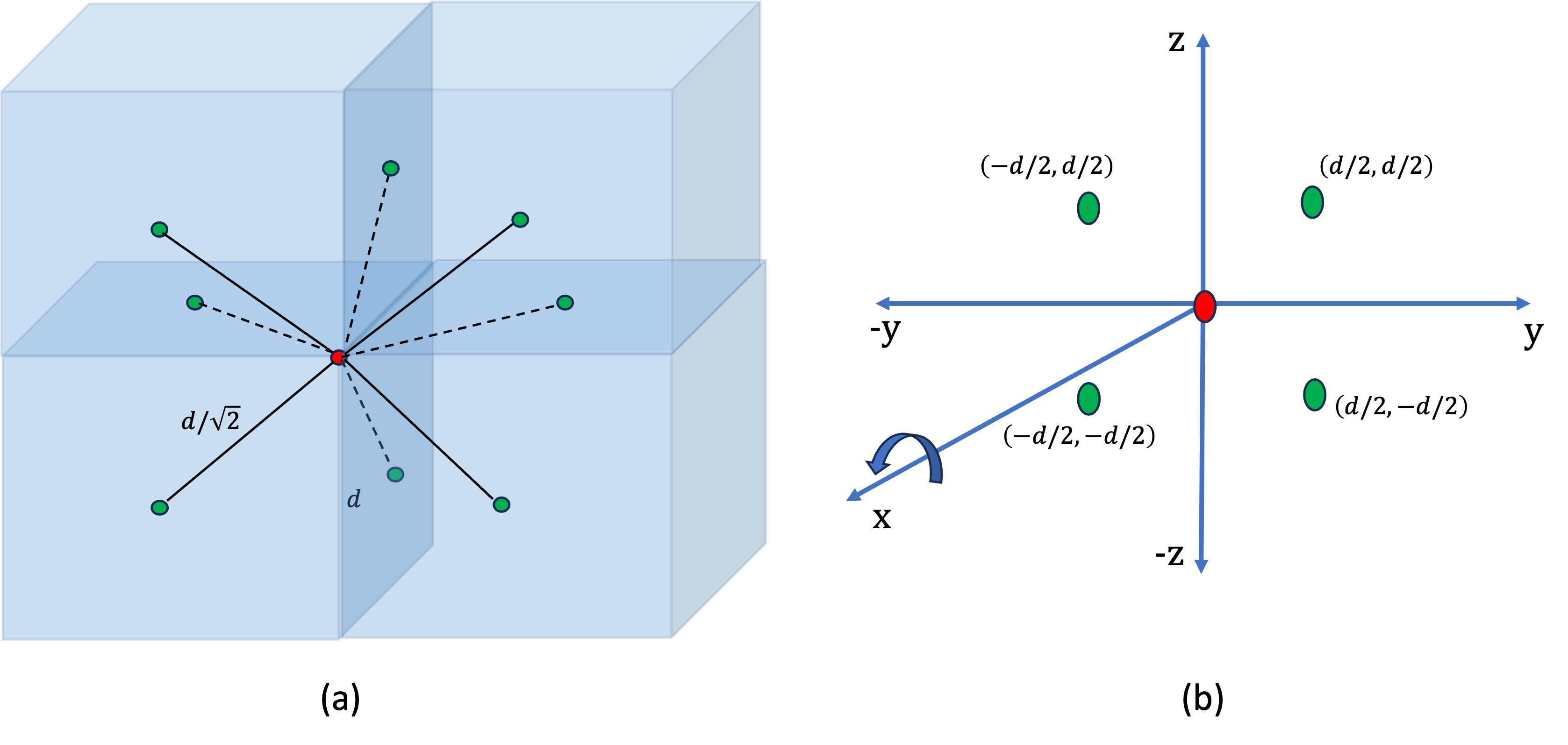} 
\caption{(a) Section of FCC lattice, (b) coordinate space that a turn could take}
\label{fig:fcc_lattice}
\end{figure}

As can be seen from the figure, the bead coloured red in the figure can have bonds with beads coloured in green. The green beads are at the centre of the faces of the cubes. The bond distance between a bead at a corner/vertex and a bead at the centre of a face is $\frac{d}{\sqrt{2}}$. Beads located at the centre of the faces of the cube can have bonds only with beads that are located at the vertices. Similarly, beads located at the vertices can have bonds with beads located at the centre of the faces. Overall, in a given plane, there are $4$ possible options to choose from. However there are $3$ planes, the $x-y$, $y-z$ and the $z-x$ planes. The turn could be in any one of those planes. There are $12$ possible directions or degrees of freedom in which the bond could turn between adjacent beads.

Now, let us have a look at the cubic lattice with planar diagonals. Fig. \ref{fig:cubic_lattice} (a) gives the section of a cubic lattice and \ref{fig:cubic_lattice} (b) gives the coordinate mapping. 
\begin{figure}[htp]
\centering 
\includegraphics[width=6in]{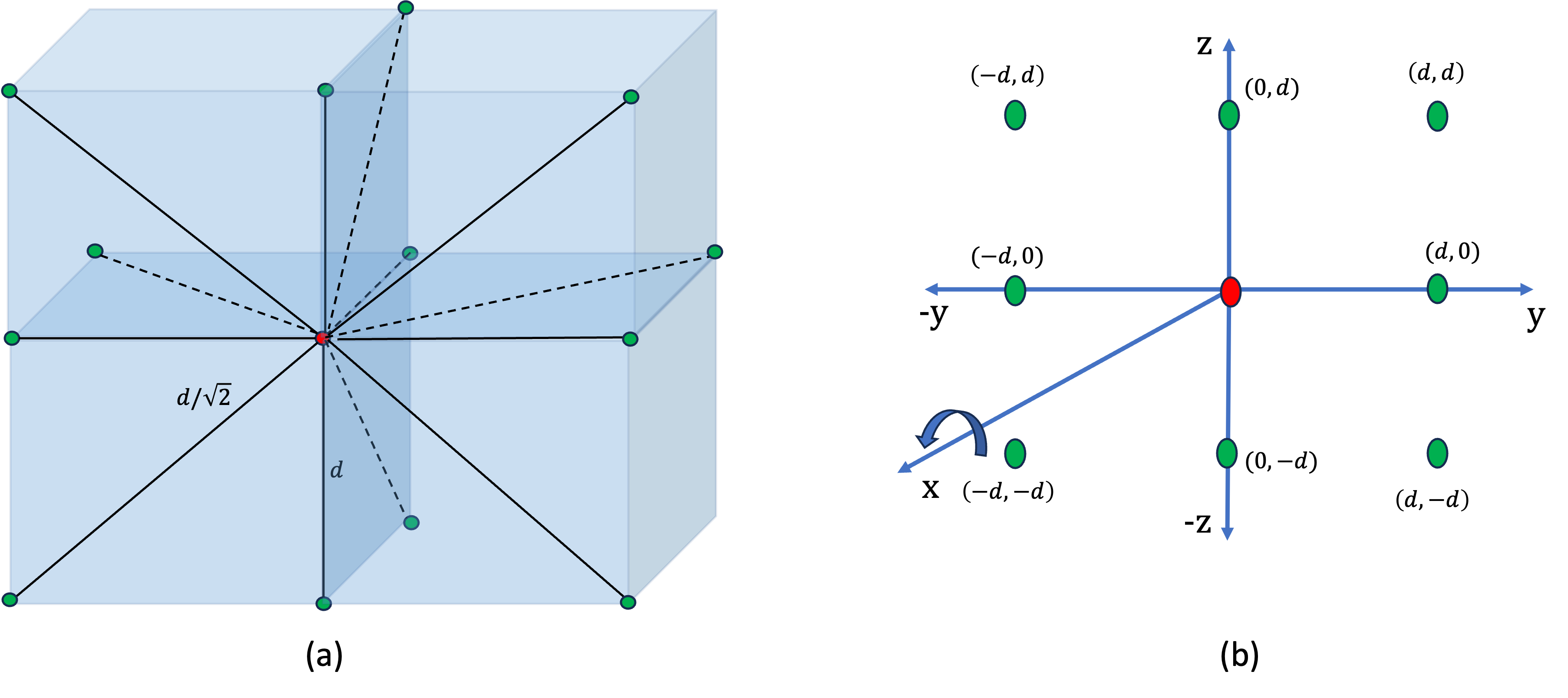} 
\caption{(a) Section of cubic lattice with planar diagonals, (b) coordinate space that a turn could take}
\label{fig:cubic_lattice}
\end{figure}
In the cubic lattice, the turns could be in $8$ directions in a given plane. With 3 planes, there are $24$ possible turns. However, there for every set of two planes, there are two directions that are repeated. For example, if we consider the $x-y$ plane and the $y-z$ plane, the $\pm y$ directions are repeated. As a result, we will have $6$ such planes that are repeated. Overall, we will have $24-6 = 18$ degrees of freedom here. 

In order to encode, we need to map the qubit states to bead coordinates. Let us consider the FCC lattice. We can start from the origin as the coordinates of the first bead. The second bead will take up a position at any one of the vertices or faces as shown in Fig. \ref{fig:fcc_lattice} (b). If we denote the axes as $x, y$ and $z$, the location of the second bead could be in any one of the $xy, yz$ or $zx$ planes. The red bead is at the origin. If the second bead were to be in the $y-z$ plane, the possible coordinates that it could have are shown in the figure. The coordinates are the $y$ and $z$ coordinates given as $(y,z)$. This is a rotation around the $x$-axis at $90^{\circ}$ interval for $4$ different points, in a counter-clockwise direction. We would see a similar set of coordinates in the $z-x$ and $x-y$ plane if the second bead were to be in the those planes. In Fig. \ref{fig:fcc_lattice} (b) the $x$ coordinates of the red and green beads are the same. We could carry forward this logic to every bead in the protein sequence. If we consider any bead $i$ with coordinates $(x,y,z)$, the bead $i+1$ would then have coordinates $(x + \Delta x,y + \Delta y,z + \Delta z)$. If the bead $i+1$ is in the plane parallel to the $y-z$ plane, $\Delta x = 0$ and $\Delta y$ and $\Delta z$ would take values as shown in Fig. \ref{fig:fcc_lattice} (b). If it is in the plane parallel to the $z-x$ plane (counter-clockwise rotation around $y$-axis),  $\Delta z$ and $\Delta x$ would take up the same values as $\Delta y$ and $\Delta z$ in the previous case, and in the same order. Similarly, when the the bead $i+1$ is in the plane parallel to the $x-y$ plane (counter-clockwise rotation around $z$-axis), $\Delta x$ and $\Delta y$ would take up similar values. 

If we look at Fig. \ref{fig:fcc_lattice}, the values that $\Delta y$ takes are $\left[\frac{d}{2}, -\frac{d}{2}, -\frac{d}{2}, \frac{d}{2} \right]$. The values that $\Delta z$ takes on the other hand are $\left[\frac{d}{2}, \frac{d}{2}, -\frac{d}{2}, -\frac{d}{2} \right]$. Let us denote the two vectors as $\Delta a$ and $\Delta b$.
\begin{equation}
    \Delta a = \left[ \begin{array}{r}
         \frac{d}{2} \\ -\frac{d}{2} \\ -\frac{d}{2} \\ \frac{d}{2} \end{array}\right] ~~~~
    \Delta b = \left[ \begin{array}{r}
         \frac{d}{2} \\ \frac{d}{2} \\ -\frac{d}{2} \\ -\frac{d}{2} \end{array}\right]   \label{delta_xyz}  
\end{equation} 

The coordinates of bead $i+1$ could be any one of the following.
\begin{flalign}
\begin{aligned}
\textrm{Parallel to the $y-z$ plane:}&~~ x_{i+1} = x_i, ~~~~~~~~~~~y_{i+1} = y_i + \Delta a_k,~~ z_{i+1} = z_i + \Delta b_k \\
\textrm{Parallel to the $z-x$ plane:}&~~ x_{i+1} = x_i + \Delta b_k,~~ y_{i+1} = y_i,~~~~~~~~~~~ z_{i+1} = z_i + \Delta a_k \\
\textrm{Parallel to the $x-y$ plane:}&~~ x_{i+1} = x_i + \Delta a_k,~~ y_{i+1} = y_i  + \Delta b_k,~~ z_{i+1} = z_i 
\end{aligned} \label{coord_update}
\end{flalign}
In the expression above, $\Delta a_k$  and $\Delta b_k$ are elements with a given index $k$ (depending upon the bond direction or turn taken) from the arrays given in (\ref{delta_xyz}).

Similarly, with the cubic lattice with planar diagonals, there are $8$ possible directions. The vectors $\Delta a$ and $\Delta b$ will have the form shown in (\ref{delta_xyz_cubic}).
\begin{equation}
    \Delta a = \left[ \begin{array}{r}
         d \\ d \\ 0 \\ -d \\ -d \\ -d \\ 0 \\ d \end{array}\right] ~~~~
    \Delta b = \left[ \begin{array}{r}
         0 \\ d \\ d \\ d \\ 0 \\ -d \\ -d \\ -d \end{array}\right]  \label{delta_xyz_cubic}  
\end{equation} 

\section{Encoding to qubit states}
We will take the example of the cubic lattice with planar diagonals to demonstrate the encoding methodology. For every plane we need to encode $8$ bits of information to encode the turn. This will be possible with, $\log_2 8 = 3$ qubits. The elements in $\Delta a$  and $\Delta b$ can be obtained by having a linear combination of the qubit states. The qubit states will be in a state of superposition and so will the elements in $\Delta a$  and $\Delta b$. As we will show in lemma \ref{lemma_qbasis}, any function $f \in \mathcal{R}^n$ can be expressed as a linear combination of product of qubit states. This mixture of qubit states forms the basis in $\mathcal{R}^N$.

\begin{lemma} \label{lemma_qbasis}
 Any function $f \in \mathcal{R}^N$ can be expressed as a linear combination of product of $n = \log_2 N$ qubit states. This set of product of qubit states will form a basis set and span the space in $\mathcal{R}^N$.
\end{lemma}

\begin{proof}
If we are given $n$ qubits, we can have $2^n$ qubit states. For example, with $3$ qubits the qubit states denoted by $QS_3$ will be as follows.
\begin{flalign}
QS_3 = \left[\begin{array}{ccc}
    q_1 & q_2 & q_3 \\
    \hline
    0 & 0 & 0 \\
    0 & 0 & 1 \\   
    0 & 1 & 0 \\
    0 & 1 & 1 \\
    1 & 0 & 0 \\
    1 & 0 & 1 \\
    1 & 1 & 0 \\
    1 & 1 & 1 
\end{array}\right] \label{QS3}
\end{flalign}
The basis that spans $\mathcal{R}^N$, can be formed by taking the product of qubit states. In an $n$-qubit system, there are ${n \choose 1}$ ways to choose one qubit, ${n \choose 2}$ ways to choose a product of two qubits, ${n \choose 3}$ ways to choose a product of three qubits, and so on till ${n \choose n}$ ways to choose product of $n$ qubits. We will also have a constant term or an intercept, which is essentially the ${n \choose 0}$ term. Together we will have 
\begin{equation*}
    {n \choose 0} + {n \choose 1} + {n \choose 2} + \hdots {n \choose n-1} + {n \choose n} = 2^n
\end{equation*}
terms. These $N=2^n$ terms will form the basis of $\mathcal{R}^N$. Any function $f \in \mathcal{R}^N$ 
can then be expressed as linear combination of these basis terms. 
\begin{equation*}
    f = \sum _{i=1} ^{N} c_i b_i
\end{equation*}
In the equation above, $b_i \in \mathcal{R}^n$ is the $i^{th}$ basis and $c_i \in \mathcal{R}$ is the coefficient. Every element of $b_i$ can be obtained by substituting the values of the qubits from a table like $QS_3$ above. Since the combinations of qubits are unique, all the basis vectors will be independent and the basis matrix $B_n = \left[ b_1, b_2, \hdots b_N \right]$ will have full rank. The coefficients can then be obtained from the equation given in (\ref{coeff_eqn}).
\begin{equation}
    c = B_n^{-1}f \label{coeff_eqn}
\end{equation}
where $c = \left[ c_1, c_2, \hdots c_N\right]^T$.
\end{proof}

To illustrate this in a $3$-qubit system, let us consider the ways we can get the product states. Let the qubits in the $3$-qubit system be $q_1, q_2$ and $q_3$. There are ${3 \choose 1} = 3$ ways of choosing one qubit, i.e., $q_1$, $q_2$ or $q_3$. There are ${3 \choose 2} = 3$ ways of choosing two qubit product states, $q_1q_2$, $q_2q_3$ and $q_3q_1$. There is ${3 \choose 3} = 1$ way of choosing 3 qubit product states, i.e. $q_1q_2q_3$. Finally, we have the constant term or the intercept term which comes from ${3 \choose 0}$. Overall, we have $8$ states each with $8$ elements. The structure of the basis matrix is given in Fig. \ref{fig:3_qubit_basis}. The first column having all $1$s correspond to the constant term.
\begin{figure}[htp]
\centering 
\includegraphics[width=2.5in]{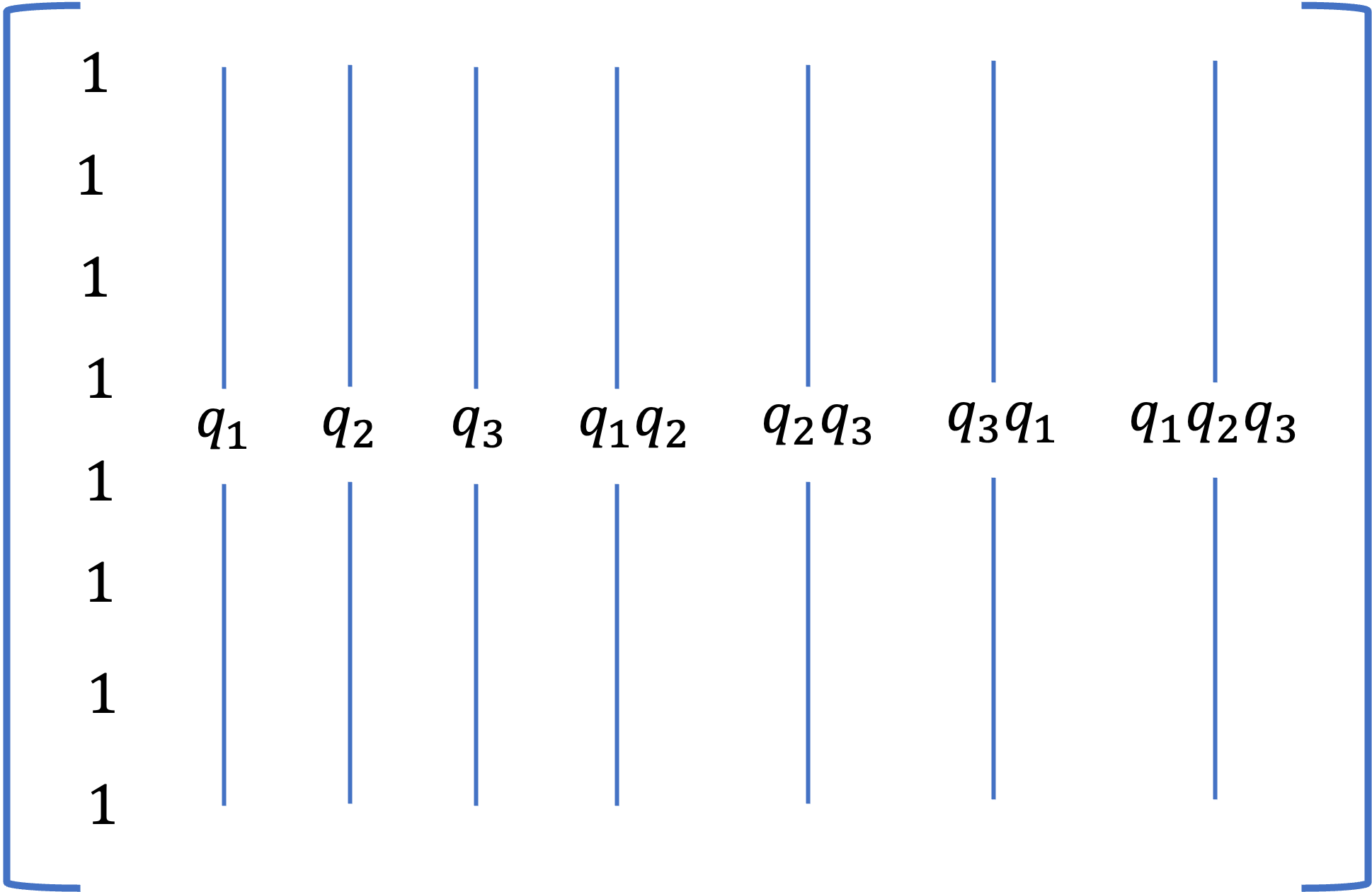} 
\caption{Basis matrix structure for a 3-qubit system}
\label{fig:3_qubit_basis}
\end{figure}

If we substitute the values of $q_1, q_2$ and $q_3$ from $QS_3$ (refer (\ref{QS3}) above) in the matrix given in Fig. \ref{fig:3_qubit_basis}, we will get the basis matrix $B_3$ as given below.
\begin{equation}
    B_3 = \left[\begin{array}{cccccccc}
         1	&	0	&	0	&	0	&	0	&	0	&	0	&	0	\\
1	&	0	&	0	&	1	&	0	&	0	&	0	&	0	\\
1	&	0	&	1	&	0	&	0	&	0	&	0	&	0	\\
1	&	0	&	1	&	1	&	0	&	0	&	1	&	0	\\
1	&	1	&	0	&	0	&	0	&	0	&	0	&	0	\\
1	&	1	&	0	&	1	&	0	&	1	&	0	&	0	\\
1	&	1	&	1	&	0	&	1	&	0	&	0	&	0	\\
1	&	1	&	1	&	1	&	1	&	1	&	1	&	1	\\
    \end{array} \right] \label{B3}
\end{equation}
We can then represent any vector in $\mathcal{R}^8$ as a linear combination of the columns in $B_3$.

\subsection{Encoding lattice structure in qubit states}
The vectors $\Delta a$ and $\Delta b$ in (\ref{delta_xyz_cubic}) can be represented using the basis vectors in $B_3$. The coefficients can be estimated from the following operations.
\begin{equation}
    \begin{aligned}
        c_{\Delta a} &= B_3 ^{-1}\Delta a \\
        c_{\Delta b} &= B_3 ^{-1}\Delta b 
    \end{aligned} \label{coeff_est}
\end{equation}
The coefficients so obtained for the cubic lattice are as follows $$c_{\Delta a} = \left[ 1.0, -2.0, -1.0, 0.0, 2.0, 0.0, -1.0, 2.0  \right]$$ and $$c_{\Delta b} = \left[ 0.0, 0.0, 1.0, 1.0, -2.0, -2.0, -1.0, 2.0 \right]$$. In other words we can express $\Delta a$ and $\Delta b$ as follows.
\begin{equation}
    \begin{aligned}
       \Delta a &= 1 - 2q_1 - q_2 + 2q_1q_2 - q_3q_1 + 2q_1q_2q_3\\
       \Delta b &= q_2 + q_3 - 2q_1q_2 - 2q_2q_3 - q_3q_1 + 2q_1q_2q_3
    \end{aligned} \label{Delta_eqn}
\end{equation}
With $\Delta a$ and $\Delta b$ we can update the coordinates of bead $i+1$ when the coordinates of bead $i$ are known. The coordinates of course will be in a state of superposition of all the options given in Fig. \ref{fig:cubic_lattice} (b).

We can do a similar exercise for the FCC lattice, where there are 4 options in a given plane. We will need $\log_2 4 = 2$ qubits. The matrix $QS_2$ to form the matrix $B_2$ is given as follows.
\begin{flalign}
QS_2 = \left[\begin{array}{cc}
    q_1 & q_2 \\
    \hline
    0 & 0 \\
    0 & 1 \\   
    1 & 0 \\
    1 & 1 \\
\end{array}\right] \label{QS2}
\end{flalign}
The basis $B_2$ will have the form shown in Fig. \ref{fig:2_qubit_basis}.
\begin{figure}[htp]
\centering 
\includegraphics[width=2.0in]{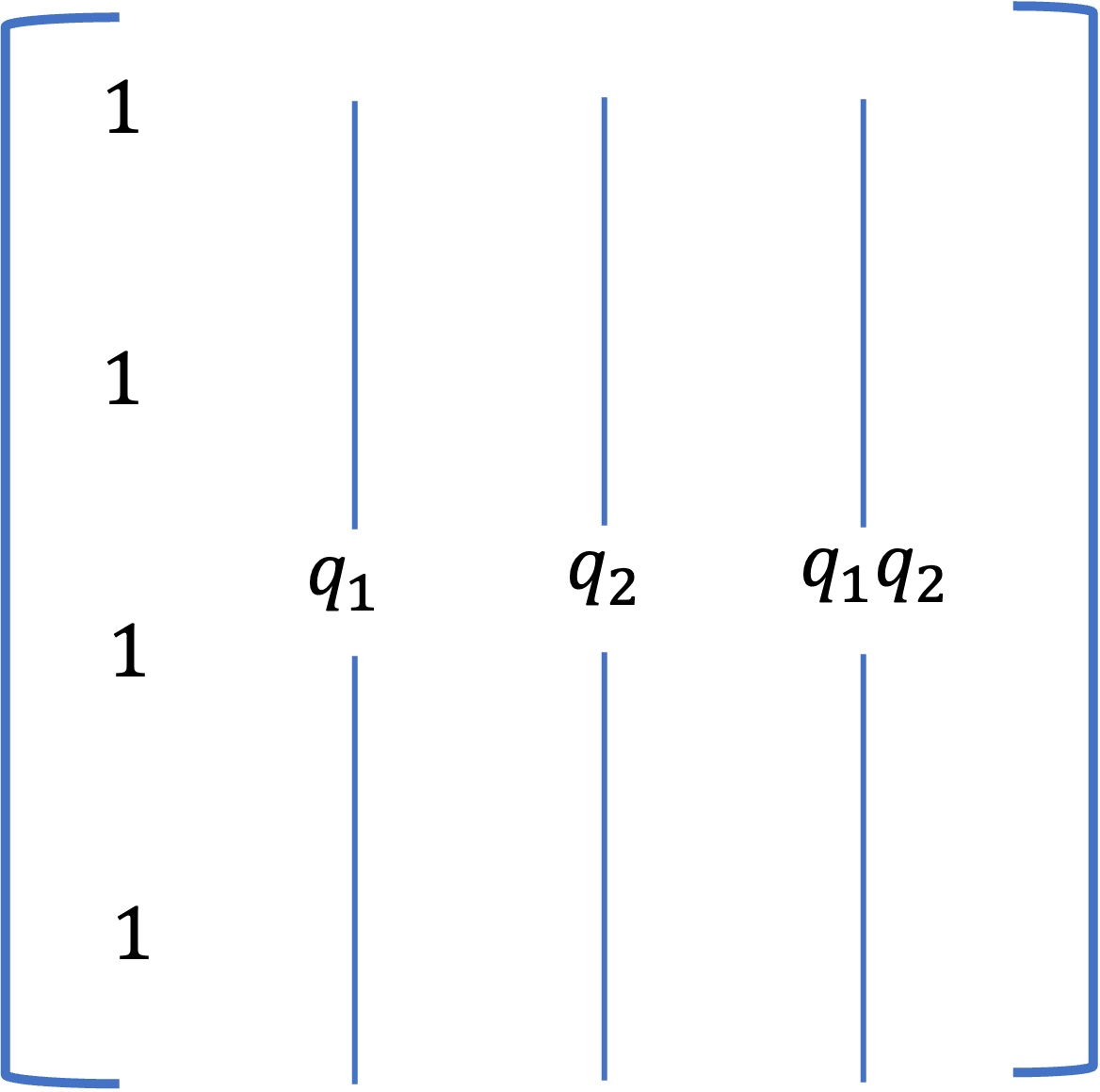} 
\caption{Basis matrix structure for a 2-qubit system}
\label{fig:2_qubit_basis}
\end{figure}
By substituting the values of the $q_1$ and $q_2$ in $QS_2$ in the matrix given in Fig. \ref{fig:2_qubit_basis}, we will get the basis matrix $B_2$ as given below.
\begin{equation}
    B_2 = \left[\begin{array}{cccc}
         1	&	0	&	0	&  0\\
1	&	0	&	1	&  0\\
1	&	1	&	0	&  0\\
1	&	1	&	1	&  1\\
    \end{array} \right] \label{B2}
\end{equation}
As in the previous case, the vectors $\Delta a$ and $\Delta b$ in (\ref{delta_xyz}) can be represented using the basis vectors in $B_2$. The coefficients can be estimated the following way.
\begin{equation}
    \begin{aligned}
        c_{\Delta a} &= B_2 ^{-1}\Delta a \\
        c_{\Delta b} &= B_2 ^{-1}\Delta b 
    \end{aligned}
\end{equation}
The coefficients obtained for the FCC lattice are as follows.
$$c_{\Delta a} = \left[ 0.5, -1.0, -1.0, 2.0  \right]$$ and $$c_{\Delta b} = \left[ 0.5, -1.0, 0, 0 \right]$$. 
We can thus express $\Delta a$ and $\Delta b$ as follows.
\begin{equation}
    \begin{aligned}
       \Delta a &= 0.5 - q_1 - q_2 + 2q_1q_2 \\
       \Delta b &= 0.5 - q_1 \\
    \end{aligned} \label{Delta_eqn}
\end{equation}

\subsection{Selection of the plane of turn}
Since there are $3$ orthogonal planes ($x-y,y-x$ and $z-x$), we need two more qubits more to make a selection of one of the planes. Let these two qubits be denoted as $q_4$ and $q_5$. The selection could be based on the computational basis states of the $q_4$ and $q_5$ as follows.
\begin{equation}
    \begin{array}{ccc}
        \textrm{Plane} & \textrm{Basis state} & \textrm{Qubit encoding}  \\
        \hline
        y-z & 01 & (1 - q_4)q_5 \\
        z-x & 10 & q_4(1 - q_5) \\
        x-y & 11 & q_4q_5
    \end{array} \label{plane_select}
\end{equation}

With the given selection criteria, we have to update equation (\ref{coord_update}) as follows.
\begin{flalign}
\begin{aligned}
    x_{i+1} &= x_i + q_4q_5\Delta a_k + q_4(1 - q_5)\Delta b_k \\
    y_{i+1} &= y_i + (1 - q_4)q_5\Delta a_k + q_4q_5\Delta b_k \\
    z_{i+1} &= z_i + q_4(1 - q_5)\Delta a_k + (1 - q_4)q_5\Delta b_k \\
\end{aligned} \label{coord_update2}
\end{flalign}

The coordinates of all the beads will be updated using equation (\ref{coord_update2}) and equation (\ref{Delta_eqn}). Overall, we will need $5$ qubits to encode a turn going from one bead to the next in a cubic lattice and $4$ qubits to encode a turn in the FCC lattice. $2$ qubits for selecting a plane and $\left \lceil{\log_2 N}\right \rceil $ qubits for encoding the bond direction (or turn) in that plane. If the peptide chain is represented by $m$ beads, we will need $(\left \lceil{\log_2 N}\right \rceil + 2)\times (m-1)$ qubits to encode the folded structure.

\subsection{Extending it to other structures}
This concept can be extended to other lattice structures as well. We need not have turns along orthogonal planes like the cubic or FCC lattices. The planes could be at an angle to the working plane. As long as the every bead or lattice point has equal and identical degrees of freedom, we could use the methodology explained here. In case, the lattice involves directions that are always at angle, we could do away with the qubits involved in plane selection and directly go for encoding directions to the computational basis states. Under such circumstances, the coordinate updates following a turn would be of the form given in (\ref{coord_upd}).
\begin{equation}
    x_{i+1} = x_i + \Delta a_k,~~ y_{i+1} = y_i  + \Delta b_k,~~ z_{i+1} = z_i + \Delta c_k 
    \label{coord_upd}
\end{equation}
The increments $\Delta a_k, \Delta b_k$ and $\Delta c_k$ can be obtained by using a variant of the equation (\ref{coeff_est}).
\begin{equation}
    \begin{aligned}
        c_{\Delta a} &= B_3 ^{-1}\Delta a \\
        c_{\Delta b} &= B_3 ^{-1}\Delta b \\
        c_{\Delta c} &= B_3 ^{-1}\Delta c 
    \end{aligned} \label{coeff_est}
\end{equation}
If there are $N$ unique directions, we could use $\left \lceil{\log_2 N}\right \rceil $ qubits to encode each of these directions.

\section{Conclusions}
In this article, we discussed a methodology to encode lattice turns to qubit computational basis states, that could find use in protein structure prediction, polymer structure study and other coarse-grained models. We showed how a combination of qubit states could be used to span the space of directions. These directions could be along planes that are orthogonal or non-orthogonal to each other. We took specific examples of cubic and FCC lattice. However, they could be extended to other lattice structures as well. 
 
\bibliographystyle{IEEEtran}
\bibliography{ref}

\end{document}